\documentclass[11pt]{article}

\usepackage{amsmath, amssymb, amsthm}
\usepackage{fullpage}
\usepackage{graphics}
\usepackage{environ}
\usepackage{framed}
\usepackage{url}
\usepackage{algorithm}
\usepackage[noend]{algpseudocode}
\usepackage[labelfont=bf]{caption}
\usepackage{cite}
\usepackage{framed}
\usepackage[framemethod=tikz]{mdframed}
\usepackage{appendix}
\usepackage{graphicx}
\usepackage[textsize=tiny]{todonotes}
\usepackage{tcolorbox}

\newcommand\remove[1]{}

\theoremstyle{plain}
\newtheorem{lemma}{Lemma}[section]
\newtheorem*{lemma*}{Lemma}
\newtheorem{corollary}[lemma]{Corollary}
\newtheorem*{corollary*}{Corollary}

\theoremstyle{definition}
\newtheorem{theorem}[lemma]{Theorem}
\newtheorem{itheorem}[lemma]{Informal Theorem}
\newtheorem*{theorem*}{Theorem}
\newtheorem{definition}[lemma]{Definition}

\newtheorem{prob}[lemma]{Problem}
\newtheorem*{rem*}{Remark}

\newif\ifrandom
\randomtrue

\newcommand{\poly}{{\rm poly}}

\author{Ofer Grossman \\ MIT \\ ofer.grossman@gmail.com \and Yang P. Liu \\ MIT \\ yangpatil@gmail.com}

\begin{document}
%\pagenumbering{gobble}

\title{Reproducibility and Pseudo-Determinism in Log-Space}

\begin{titlepage}

\clearpage\maketitle
\thispagestyle{empty}

\begin{abstract}
A curious property of randomized log-space search algorithms is that their outputs are often longer than their workspace. This leads to the question: how can we reproduce the results of a randomized log space computation without storing the output or randomness verbatim? Running the algorithm again with new random bits may result in a new (and potentially different) output. 

We show that every problem in search-RL has a randomized log-space algorithm where the output can be reproduced.  Specifically, we show that for every problem in search-RL, there are a pair of log-space randomized algorithms $A$ and $B$ where for every  input $x$, $A$ will output some string $t_x$ of size $O(\log n)$, such that $B$ when running on $(x, t_x)$ will be pseudo-deterministic: that is, running $B$ multiple times on the same input $(x, t_x)$ will result in the same output on all executions with high probability. Thus, by storing only $O(\log n)$ bits in memory, it is possible to reproduce the output of a randomized log-space algorithm.

An algorithm is reproducible without storing \textit{any} bits in memory (i.e., $|t_x|=0$) if and only if it is pseudo-deterministic. We show pseudo-deterministic algorithms for finding paths in undirected graphs and Eulerian graphs using logarithmic space. Our algorithms are substantially faster than the best known deterministic algorithms for finding paths in such graphs in log-space.

The algorithm for search-RL has the additional property that its output, when viewed as a random variable depending on the randomness used by the algorithm, has entropy $O(\log n)$.

\end{abstract}

\end{titlepage}

\newpage

\section{Introduction}

\subsection{Reproducible Outputs}
When using a log-space machine to perform a randomized search algorithm with a polynomial-sized output, the output cannot be fully stored. Running the algorithm again with new random bits may result in a new (and potentially different) output. Hence, after running the computation, we lose access to the outputted answer, and are unable to reproduce it.

Consider, for example, the following simple computational problem: Given a (directed) graph $G$ and two vertices $s$ and $t$ such that a random walk from $s$ hits $t$ quickly with high probability, output two copies of the same path from $s$ to $t$. That is, the goal is to output some path $P$, and then output the same path $P$ again. It's not clear how to perform the above in randomized log-space, since after outputting some path $P$, it's not clear how to reproduce $P$ and be able to output it again. So, although outputting a single path is easy, or two potentially different paths, it's not clear how to output the same path twice.

Another example of this phenomenon in play is that it is known that there is a randomized reduction from NL to UL (in fact, NL is reducible to UL $\cap$ coUL) \cite{RA, GW}. It follows that if UL $\subseteq$ RL, then NL can be solved by randomized log-space algorithms with \textit{two-way access} to the random bits (that is, there is a randomized disambiguation of $NL$ which uses two-way access to the random bits). However, when assuming UL $\subseteq$ RL, it is not known whether NL can be solved by a randomized log-space algorithm with \textit{one-way} access to the random bits. The two-way access to the random bits is needed so that the output of the reduction (which is an instance of a problem in UL) can be accessed in a two-way fashion. If the output of the reduction was reproducible, then one-way access to the random bits would suffice.

One way to achieve reproducibility is through \textit{pseudo-determinism}. Pseudo-deterministic algorithms are randomized search algorithms which, when run on the same input multiple times, with high probability output the same result on all executions. Given such an algorithm, it is possible to reproduce outputs: simply run the algorithm again using new randomness. We manage to achieve reproducibility using a different and novel approach which does not involve finding a pseudo-deterministic algorithm for the problem.

\subsection{Our Contribution}

Our contribution falls into two parts: contributions to reproducibility in the context of log-space, and contributions to pseudo-determinism in the context of log-space.

\paragraph{Reproducibility:} We introduce the notion reproducibility and provide a definition in Section \ref{sec:alg}. Our main result shows that every problem in search-RL (see Section \ref{sec:prelim} for a definition of search-RL) can be solved so that its output is reproducible. By reproducible, we essentially mean that the algorithm will be able to generate many identical copies of its output using only $O(\log n)$ memory. Then, it effectively has two-way access to the output, instead of only one-way access, which is the case with standard search-RL algorithms.

In order to achieve reproducibility, we show that for every problem in search-RL there is some randomized log-space algorithm $A$ such that with high probability, the output of $A$ only depends on the first $O(\log n)$ random bits $A$ samples. That is, after sampling the first $O(\log n)$ random bits, with high probability for most choices of the rest of the random bits used by the algorithm, the same result will be outputted. This property allows the algorithm reproduce the output, as it can store the first $O(\log n)$ random bits it sampled in memory, and use them to recreate the output. Since the algorithm can find and store the information needed to reproduce the answer, we say that the output is \emph{reproducible}.

Our first result is that every problem in search-RL (as defined in \cite{RTV}) has a randomized log-space algorithm whose output, with high probability, only depends on its first $O(\log n)$ random bits. This implies that every problem in search-RL has reproducible solutions.

\begin{itheorem}
\label{infpsdlintro}
Every problem in search-RL has a randomized log-space algorithm whose output, with high probability, only depends on its first $O(\log n)$ random bits.
\end{itheorem}

A more precise statement is given in Section \ref{sec:alg} as Theorem \ref{psdlintro}. The algorithm we present has several other noteworthy properties, which we discuss in Subsection \ref{entropy}. This includes that the output of the algorithm, when viewed as a random variable depending on the random choices used by the algorithm, has entropy $O(\log n)$. This is significantly lower than a standard search-RL algorithm, which may have polynomial entropy.

%Our algorithm has an additional property: if one fixes the first $O(\log n)$ random bits sampled by the algorithm, the algorithm will output the same solution with high probability over the rest of the random bits. This property implies that all random bits other than the first $O(\log n)$ random bits still ``look random'' to an outside observer, even after seeing the algorithm's output. That is, the output leaks very little information about those random bits. Hence, if one adaptively queries our algorithm $k$ times in a row, our algorithm can reuse all but $O(\log n)$ of the random bits for future queries, without a substantial decrease in success probability.

\paragraph{Pseudo-determinism:} In later sections, we show faster pseudo-deterministic algorithms for finding paths in undirected and Eulerian graphs. These algorithms are reproducible even without storing $O(\log n)$ bits in memory.

For undirected graphs, a deterministic log-space algorithm has been shown by Reingold \cite{reingold}. One of the drawbacks of this algorithm is that its runtime, while polynomial, has a very large exponent, since it requires going over all paths of length $O(\log n)$ on a certain graph, with a large constant hidden in the $O$ (for certain expositions of the algorithm, the polynomial runtime is larger than $O(n^{10^9})$). This can likely be improved, but we imagine it would be difficult to lower it to a ``reasonable'' polynomial time complexity. We show a pseudo-deterministic algorithm for the problem which runs in the more reasonable time of $\tilde{O}(mn^3)$:

\begin{theorem}[Pseudo-deterministic Undirected Connectivity in $\tilde{O}(mn^3)$ time, $O(\log n)$ space]
\label{slpsd}
Let $G$ be a given undirected graph with $n$ vertices and $m$ edges. Given two vertices $s$ and $t$ of $G$ which are connected, there is a pseudo-deterministic log-space algorithm which outputs a path from $s$ to $t$. Furthermore, the algorithm runs in time $\tilde{O}(mn^3)$.
\end{theorem}

We then generalize the theorem to Eulerian graphs (directed graphs where each vertex has indegree equal to its outdegree). Finding paths in such graphs deterministically has been shown in \cite{RTV}. Once again, the algorithm given in \cite{RTV} suffers from a very large polynomial runtime.

\begin{theorem}[Connectivity in Eulerian graphs in in $\tilde{O}(m^5n^3)$ time, $O(\log n)$ space]
\label{eulerpsd}
Let $G$ be a given Eulerian graph with $n$ vertices and $m$ edges. Given two vertices $s$ and $t$ of $G$ such that there is a directed path from $s$ to $t$, there is a pseudo-deterministic log-space algorithm which outputs a path from $s$ to $t$. Furthermore, the algorithm runs in time $\tilde{O}(m^5n^3)$.
\end{theorem}

\subsection{Related Work}
%preserving randomness for adaptive algorithms
%impagliazzo zuckerman reusing random bits
%saks zhou
%reingold stuff
%RTV eulerian grpahs
\paragraph{Pseudo-determinism:}
The study of pseudo-determinism was initiated by Gat and Goldwasser \cite{GG}. Pseudo-deterministic algorithms have been studied for the problem of finding primitive roots modulo primes \cite{roots}, for finding perfect perfect matchings in parallel on bipartite graphs \cite{matching}, in the context of sublinear time algorithms \cite{GGR}, and in the context of interactive proofs \cite{proofs}. Furthermore, general theorems regarding the existence of pseudo-deterministic subexponential algorithms have been shown in \cite{OS, dhiraj}. In this work, we initiate the study of pseudo-determinism in the context of low space computation.

\paragraph{RL vs L:}
Related to our result on pseudo-deterministic undirected connectivity is the work of Reingold, which showed that undirected connectivity can be solved \textit{deterministically} with logarithmic space \cite{reingold}. Later, this result was extended to find pseudo-random walks on Eulerian graphs by Reingold, Trevisan, and Vadhan \cite{RTV}.

One of our techniques may remind some readers of the work of Saks and Zhou that show that problems in BPL can be solved deterministically using $O(\log^{3/2}(n))$ space \cite{SZ}. In \cite{SZ}, the authors add random noise to certain computed matrices, in order to be able to reuse certain random bits. In this work, we pick a certain `threshold' at random, and this allows us to reuse randomness (more accurately, it makes our output be pseudo-deterministic with respect to certain random bits). The two ideas are similar in that they use randomization in an unconventional way in order to make the output not depend on certain random bits (for Saks and Zhou, this was helpful since it allowed those random bits to be reused).

\section{Preliminaries}\label{sec:prelim}

In this section we establish some definitions and lemmas that will be useful in later parts of the paper. Many of our definitions, especially those related to search problems in the context of log-space, follow closely to the definitions in \cite{RTV}.

We begin by defining a search problem.
\begin{definition}[Search Problem]
A \textit{search problem} is a relation $R$ consisting of pairs $(x, y)$. We define $L_R = \{x | \exists y$ s.t. $(x, y) \in R\}$, and $R(x) = \{y|(x, y) \in R\}$.
\end{definition}
The computational task associated with a search problem $R$ is: given $x$, find a $y$ such that $(x, y) \in R$. From this point of view, $L_R$ corresponds to the set of valid inputs, and $R(x)$ corresponds to the set of valid outputs on input $x$.

We now define a pseudo-deterministic algorithm. Intuitively speaking, a pseudo-deterministic algorithm is a randomized search algorithm which, when run multiple times on the same input (using different random strings), results in the same output with high probability.
\begin{definition}[Pseudo-deterministic]
A randomized search algorithm $A$ is \textit{pseudo-deterministic} if for all valid inputs $x$,
$$\Pr_{r_1, r_2}(A(x, r_1) = A(x, r_2)) \ge 2/3.$$
\end{definition}
We note that through repetition, the $2/3$ in the above definition can be amplified.

We now define classes of search problems in the context of log-space. Our definitions follow closely to those of \cite{RTV}.

\begin{definition}[Log-space search problem]
A search problem $R$ is \textit{log-space} if there is a polynomial $p$ such that if $y \in R(x)$ then $|y| \le p(|x|)$ and there is a deterministic log-space machine can decide if $(x, y) \in R$ with two-way access to $x$ and one-way access to $y$.
\end{definition}

We now define the class search-L. We remind the reader that a transducer is a Turing machine with a read-only input tape, a work tape (in our case, of logarithmic size), and a write-only output tape.

\begin{definition}[search-L]
A search problem $R$ is in \textit{search-L} if it is log-space and if there is a logarithmic space transducer $A$ such that $A(x) \in R(x)$ for all $x$ in $L_R$.
\end{definition}

\begin{definition}[search-RL]
A search problem $R$ is in \textit{search-RL} if it is log-space and if there is a randomized logarithmic space transducer $A$ and polynomial $p$ such that $\Pr_r[A(x, r) \in R(x)] \ge \frac{1}{p(|x|)}$ for all $x \in L_R$.
\end{definition}

The following computational problem is complete for search-RL:

\begin{definition}[\textsc{Short-Walk Find Path}]
\label{swfp}
Let $R$ be the search problem whose valid inputs are $x = (G, s, t, 1^k)$ where $G$ is a directed graph, $s$ and $t$ are two vertices of $G$, and a random walk of length $k$ from $s$ reaches $t$ with probability at least $1 - 1/|x|$ (where $|x|$ represents the length of the input $x$). On such an $x$, a valid output is a path of length up to $\poly(k)$ from $s$ to $t$.
\end{definition}

%Note that we can amplify $\frac{1}{2}$ to any constant less than $1$.
\begin{lemma}
\label{completesearchrl}
\textsc{Short-Walk Find Path} is complete for search-RL.
\end{lemma}

We prove the above lemma in Appendix \ref{complete}, via a reduction from \textsc{Poly-Mixing Find Path}, which was shown to be complete for search-RL in \cite{RTV}.

Before going on to the algorithm in Section \ref{sec:alg}, we make a definition to simplify the explanations.

\begin{definition}\label{def:pk} For a graph $G$ with vertices $s$ and $t$, and a positive integer $k$, let $p_k(s, t)$ denote the probability that a random walk of length $k$ starting from $s$ goes through $t$.
\end{definition}

One of the key lemmas used by the algorithm is that one can estimate the value of $p_k(s, t)$ up to some polynomial additive error in search-RL. To do so, we simulate polynomially many random walks starting at $s$ and count the fraction that pass through $t$. This is made precise in the following lemma:

\begin{lemma}
\label{estprob}
Consider a graph $G$ with $n$ vertices, two of which are $s$ and $t$. Let $k$ be a positive integer. Then there exists a randomized log-space algorithm that on input $(G, s, t, 1^k)$ outputs an estimate $\mu$ for $p_k(s, t)$ satisfying $|\mu - p_k(s, t)| \le \frac{1}{k^5n^5}$ with probability at least $1 - 2e^{-2kn}$.
\end{lemma}

\begin{proof}
To find $\mu$, we simulate many random walks from $s$ of length $k$, and then output the fraction which reach $t$. More precisely, we use the following algorithm: simulate $k^{11}n^{11}$ random walks of length $k$ starting at $s$, and count how many end at $t$. Say that $C$ of them do. Then output $\frac{C}{k^{11}n^{11}}.$ To show that this works, it suffices to note that \[ \Pr\left[\left|\frac{C}{k^{11}n^{11}} - p_k(s, t)\right| \le \frac{1}{k^5n^5} \right] \ge 1 - 2e^{-2kn} \] by Hoeffding's inequality.
\end{proof}

\section{An Algorithm for Search-RL with Reproducible Outputs}\label{sec:alg}

\subsection{Reproducibility}

We begin with a formal definition of a problem with reproducible outputs. Essentially, a problem has reproducible solutions if for every input $x$ we can generate a short string $t_x$ so that given both $x$ and $t_x$ we can keep reproducing copies of the same $y$ satisfying $(x, y) \in R$. That is, by memorizing only the short string $t_x$, we can continue to produce more copies of the same output.

\begin{definition}[Reproducible]\label{reproducible}
We say that a search problem $R$ has \textit{log-space reproducible} solutions if there exist randomized log-space algorithms $A$ and $B$ satisfying the following properties:

\begin{itemize}
\item On input $x$, with high probability, $A$ outputs a string $t_x$ of length $O(\log n)$ such that the second bullet holds.\footnote{Via exhaustive search, it can be shown that if for all inputs $x$ there exists a $t_x$ such that the second bullet holds, there also exists a log-space algorithm $A$ that for all $x$ with high probability will output some $t_x$ satisfying the second bullet.}
\item There exists some $y$ satisfying $(x, y) \in R$ such that with high probability $B$ outputs $y$ when running on input $(x, t_x)$.
\end{itemize}

\end{definition}

Reproducibility is closely related to pseudo-determinism. In the case where $t_x$ is of size 0, the algorithm $B$ is a pseudo-deterministic algorithm for the search problem $R$.

An alternate way to view reproducibility is that a search problem $R$ has log-space reproducible solutions if there exists some randomized log-space algorithm $C$ such that algorithm $C$ can produce two copies of an output $y$ satisfying $(x, y) \in R.$ Essentially, this alternate view captures the fact that a problem has log-space reproducible solutions if and only if we can produce some output, and then produce it again, ensuring that the output was not lost after the first time we computed it.

\begin{lemma}
\label{copies}
A search problem $R$ has log-space reproducible solutions if and only if there exists some randomized log-space algorithm $C$ such that for all valid inputs $x$ (i.e., $x \in L_R$), with high probability $C(x)$ outputs two copies of an output $y$ satisfying $(x, y) \in R$. That is, with high probability $C$ outputs the tuple $(y, y)$, where $(x, y) \in R$.
\end{lemma}

\begin{proof}

First we show that every search problem $R$ with log-space reproducible solutions has a randomized log-space algorithm $C$ that given some $x$, with high probability outputs two copies of an output $y$ satisfying $(x, y) \in R$. Let $A$ and $B$ be the algorithms for problem $R$ from Definition \ref{reproducible}. We now show how to amplify algorithm $B$  so that the probability it outputs $y$ is $1 - \frac{1}{4n^2}$. We do this by determining the $i$-th bit of the output for all $1 \le i \le |y|, $ where $|y|$ denotes the length of the output $y$. More specifically, consider the algorithm $B'$ that loops through all $i$ such that $1 \le i \le |y|$, and for each index $i$, runs $B$ at least $\Omega(\log(2n|y|))$ times to determine the most common bit in that position. Because there exists an input $y$ such that $\Pr_r[B(x, t_x, r) = y] \ge \frac23$, the most common bit in each position will be the same as the bit of $y$ in that position. Therefore, (after choosing a large enough constant in the $\Omega$) by a Chernoff bound and a union bound over all bits in $y$, $B'$ will output $y$ with probability at least $1 - \frac{|y|}{8n^3|y|^3}.$ Now, an algorithm $C$ for $R$ can do the following: first run $A$ to get $t_x$, and then run algorithm $B'$ two times. By a union bound, with high probability, the output will be $y$ both times, as the failure probability is bounded by $\frac{2|y|}{8n^3|y|^3},$ so the success probability is high.

Now we show the reverse direction. Consider an algorithm $C$ such that with high probability on input $x$, $C$ will output two copies of an output $y$ satisfying $(x, y) \in R.$ Now we construct algorithms $A$ and $B$ satisfying the conditions of Definition \ref{reproducible}. First, let algorithm $A$ simulate algorithm $C$, and output the state of the Turing machine corresponding to algorithm $C$ after $C$ has outputted one copy of $y$ (that is, after it has otputted the comma between the two $y$'s in $(y, y)$). This will be our string $t_x.$ The length of $t_x$ will be of size $O(\log n)$ as $C$ is a log-space algorithm. Now, algorithm $B$ will continue simulating algorithm $C$, starting from state $t_x.$ With high probability, $C$ will output another copy of $y$ after reaching state $t_x$, since we know which high probability $C$ outputs the pair same output twice. Therefore, algorithm $B$ will output $y$ with high probability, as desired.
\end{proof}

We note that if a problem $R$ has reproducible solutions, then for any polynomially bounded $\ell$ it has a randomized logspace algorithm $D$ which on valid input $x$ outputs $\ell$ copies of a valid output $y$. That is, it outputs $(y, y, \ldots, y)$, where $(x, y) \in R$. This can be done by first running the algorithm $A$ (from Definition \ref{reproducible}) to create an advice string $t_x$, and then running algorithm $B$ (from Definition \ref{reproducible}) $\ell$ times using the same advice string $s$ on all those $\ell$ executions.

\paragraph{A justification for the definition of reproducibility:} One may argue that the definition proposed for reproducibility is highly structured, and that there may be ways to construct algorithms which capture the notion of reproducibility without adhering to the structure of Definition \ref{reproducible}. To argue that Definition \ref{reproducible} captures the ``true" notion of reproducibility, we can consider a ``weak'' definition of reproducibility, and show that it is equivalent to the strict Definition \ref{reproducible}. Since the ``weak'' and strict versions of the definition can be shown to be equivalent, we know that we must have captured the notion of reproducibility which must lie between those two definitions.

We note that the ``weakest'' possible notion for reproducibility is that an algorithm has reproducible solutions if there is some randomized logspace algorithm $C$ which outputs the same valid output $y$ twice. This is since if such an algorithm doesn't exist, then there is no hope to achieve any sort of reproducibility, since we essentially can't even reproduce the output a single time. Lemma \ref{copies} shows that this weak notion is equivalent to Definition \ref{reproducible}, showing that the seemingly too-strict definition (Definition \ref{reproducible}) is equivalent to the seemingly too weak definition (that a problem has reproducible outputs if there exists a $C$ which outputs the same valid output $y$ twice), demonstrating that the correct notion of reproducibility is captured by Definition \ref{reproducible}.

\subsection{Algorithms with few influential bits}

To construct a log-space algorithm whose output is reproducible, we will design an algorithm $A$ whose output depends on $O(\log n)$ of the random bits $A$ samples. Then, the algorithm can store those $O(\log n)$ influential random bits, and using those it can reproduce its output by running again using the same $O(\log n)$ influential random bits. Below we give a precise definition of what we mean by ``influential random bits''.

\begin{definition}[Influential bits]\footnote{An algorithm is pseudo-deterministic if it has zero influential bits. In this sense, the above definition is an extension of pseudo-determinism.}
\label{infl}
Let $k(n)$ be a polynomial-time computable function. Say that a randomized log-space search algorithm $A$ has \emph{$k(n)$ influential bits} if for all valid inputs $x$, with probability at least $\frac12$ over random strings $r_1$ of length $k(n)$, we have that there exists an output $y$ such that $y$ is valid for input $x$ and $\Pr_{r_2}[A(x, r_1, r_2) = y] \ge \frac23.$ Here, $r_2$ denotes the remaining randomness (after $r_1$) used by $A$ and $A(x, r_1, r_2)$ denotes the output of $A$ on input $x$ with randomness $r_1$ and $r_2$.
\end{definition}

We now prove that if a randomized log-space algorithm $A$ has $O(\log n)$ influential bits, then its output is reproducible. Essentially, the idea is that the algorithm $A$ can store its $O(\log n)$ influential random bits in memory and then use these bits to recompute its previous output.

\begin{lemma}\label{bitstoreproduce}
If a search problem $R$ can be solved by a randomized log-space algorithm with $O(\log n)$ influential random bits, then it has log-space reproducible solutions.
\end{lemma}

\begin{proof}
Let $C$ be an algorithm for the search problem $R$ with $b = O(\log n)$ influential random bits. Let $m = \poly(n)$ be an upper bound on the output size. We will construct algorithms $A$ and $B$ that satisfy the conditions of Definition \ref{reproducible}, i.e., for an input $x$, $A$ outputs a string $t_x$ of length $O(\log n)$, such that with high probability algorithm $B$ running on input $(x, t_x)$ will pseudo-deterministically produce an output $y$ satisfying $(x, y) \in R.$

First, we will show how to amplify the $\frac{2}{3}$ from Definition \ref{infl}. We randomly generate $b$ bits (recall that $b$ is the number of influential random bits used by algorithm $C$). With probability at least $\frac12$, fixing these $b$ bits will cause algorithm $C$ to produce the same output for at least $\frac23$ of the choices for the remaining random bits. We can amplify the $\frac23$ via repetition. That is, we can create a new algorithm $C'$ where the $k$th output bit is $0$ if after running $C$ a total of $cn$ times (for some constant $c$), the majority of times the $k$th output bit was a 0. Otherwise, we set the $k$th output bit to be a $1$. If we let $c$ be sufficiently large, then by a Chernoff bound and a union bound over the coordinates of the output, we have that the whole output of $C'$ will be $y$ with probability at least $1 - \frac{1}{2^n}.$

Now, we describe the algorithms $A$ and $B$. Algorithm $A$ begins by sampling a random string $s_1$ of length $b$ bits. Next, algorithm $A$ will test whether $s_1$ is a ``good'' string. That is, we test whether with high probability there is some $y$ such that the probability $\Pr_{r_2}[C'(x, s_1, r_2) = y]$ is large (at least $1 - \frac{1}{n^2}$). This can be done by, for each output bit $i$, running the algorithm $C'$ a total of $\Theta(n^2)$ times, and checking if the $i$-th output bit was the same in all executions (we remark here that $A$ knows which bits of $C'$ are influential because those are the bits that are sampled first). If $\Pr_{r_2}[C'(x, s_1, r_2) = y]$ is at least $1 - \frac{1}{2^n}$, then $s_1$ passes this test with probability at least $1 - \frac{\Theta(n^2)}{2^n}.$ If the string $s_1$ passes, we know that with high probability $\Pr_{r_2}[C'(x, s_1, r_2) = y] \ge 1 - \frac{1}{n}$. If it is the case that for each coordinate, $C'$ outputted the same bit on each of the executions, algorithm $A$ can output $s_1$ as its string $t_x$. Otherwise, if one of the output bits was not the same on all executions (i.e., $s_1$ did not pass the test for having a high value of $\Pr_{r_2}[C'(x, s_1, r_2) = y]$), we sample a new string $s_1$ and repeat. After $O(\log n)$ tries for the string $s_1$, with high probability we will find a good string $s_1$, where $C'$ outputs a certain $y$ with high probability. Now, algorithm $B$ can simply simulate algorithm $C'$ on the input $(t_x, x)$, where $t_x$ is the good string that $A$ outputted. Since with high probability $\Pr_{r_2}[C'(x, s_1, r_2) = y] \ge 1 - \frac{1}{n}$, we know that algorithm $B$, when run multiple times on $(t_x, x)$, will output the same $y$ with high probability.
\end{proof}

In the rest of the section, we prove that every problem in search-RL has an algorithm with $O(\log n)$ influential bits:

\begin{theorem}
\label{psdlintro}
Every problem in search-RL has a randomized log-space algorithm that only has $O(\log n)$ influential bits.
\end{theorem}

As an immediate corollary of Theorem \ref{psdlintro} and Lemma \ref{bitstoreproduce} we have:

\begin{corollary}
Every problem in search-RL has log-space reproducible solutions.
\end{corollary}

%To make the notion of entropy in randomized algorithms more precise, we make the following definition.

%\begin{definition} Let PSDL[$g(n)$] denote the class of search problems that can be solved with a log-space randomized algorithm whose output has entropy $O(g(n)).$ \end{definition}

%With the above defintion in mind, we can restate Theorem \ref{psdlintro} as searchRL $\in$ PSDL[$\log n$]. Now, we are almost ready to present an overview of the algorithm for the problem in Lemma \ref{completesearchrl}.

\subsection{High Level Proof Idea for Theorem \ref{psdlintro}}

At the high level, the idea for the algorithm for Theorem \ref{psdlintro} is as follows. First, we consider the problem \textsc{Short-Walk Find Path} from Definition \ref{swfp}, which we know is complete for search-RL by Lemma \ref{completesearchrl}. Now, suppose that we wish to find a path from $s$ to $t$, and we know that $p_k(s, t) \ge \frac{1}{2}$ (see Definition \ref{def:pk} for a definition of $p_k$). This implies that there must exist an outneighbor $v$ of $s$ such that $p_{k-1}(v, t) \ge \frac{1}{2}.$ Therefore, if we could estimate $p_{k-1}(v, t)$ for all outneighbors $v$ of $s$, we could pick the lexicographically first neighbor satisfying $p_{k-1}(v, t) \ge \frac{1}{2}$, and continue recursively from there. Since $v$ is uniquely determined (it is the lexicographically first outneighbor of $s$ satisfying $p_{k-1}(v, t) \ge \frac{1}{2}$), if such an algorithm worked, it would be fully pseudo-deterministic (and hence would have no influential random bits).

Unfortunately, this proposed algorithm of finding an outneighbor will not work. To see why, consider the situation where for the first outneighbor $v$ of $s$ that we check, $p_{k-1}(v, t)$ is exactly equal to $1/2$. Then no matter how accurately we estimate $p_{k-1}(v, t)$, much of the time our estimate will be less than $1/2$, and other times it will be greater than $1/2$. This makes our algorithm not pseudo-deterministic, as in some runs we will use vertex $v$ in the path, and in other runs we will not.

Instead, we construct an algorithm that has logarithmically many influential bits in the following way. We will generate a \emph{threshold} $c$ (from some distribution) and find the first outneighbor $v$ satisfying $p_{k-1}(v, t) \ge 1/2 - c$, and use that vertex $v$ as part of our path. Then, we recurse to find the next vertex in the path. Of course, this still fails if $p_{k-1}(v, t) = 1/2 - c$ (or if $p_{k-1}(v, t)$ is close to $1/2 - c$). However, if the value of $c$ is far away from all values of $1/2 - p_{i}(u, t)$ for all $u$ and all $1 \le i \le k$, then our algorithm, for this fixed value of $c$, will always give the same output. Hence, to get an algorithm with logarithmically many influential bits for \textsc{Short-Walk Find Path}, we just need a way to use logarithmically many bits to select a value of $c$ such that for all $1 \le i \le k$, and vertices $v$, $|1/2 - p_{k-1}(v, t) - c|$ is large (at least $1 / n^5k^5$).

We are able to find such a value of $c$ by sampling it at random from some set of polynomial size. Note that there are $kn$ possible values for an expression of the form $1/2 - p_{i}(u, t)$ (with $i \le k$), since there are $n$ options for $u$, and $k$ options for $i$, and we need $c$ to be far from all $kn$ of these options. If we were to randomly sample $c$ from the set $\{\frac{1}{k^4n^4}, \frac{2}{k^4n^4}, \dots, \frac{k^2n^2}{k^4n^4} \}$, then with high probability our chosen value of $c$ would be far away from all of the expressions of the form $1/2 - p_{i}(u, t)$ (with $i \le k$), and hence once we fix such a $c$, we get the same output with high probability. Because we can sample $c$ using $O(\log nk)$ bits, our output will only depend on the first $O(\log nk)$ bits sampled.

\subsection{Algorithm and Analysis}

Here we will state the algorithm for Theorem \ref{psdlintro} more precisely and provide a detailed analysis.

\begin{algorithm}
\caption{Randomized algorithm with $O(\log n)$ influential bits for \textsc{Short-Walk Find Path} on input $(G, s, t, 1^k)$}
\begin{algorithmic}[1]

%\Statex
\State Initialize $u = s$. $u$ is the current vertex.
\State Choose a threshold $c$ from the set $\{\frac{1}{k^4n^4}, \frac{2}{k^4n^4}, \dots, \frac{k^2n^2}{k^4n^4} \}$ uniformly at random.
\For{$d = k, k-1, \dots, 1$}
	\State Print $u$ (on the output tape).
    \For{each outneighbor $v$ of $u$ (in lexicographic order)}
    	\State Estimate $1/2 - p_{d-1}(v, t)$, up to additive error $\frac{1}{k^5n^5}$ (use Lemma \ref{estprob}). Call the estimate $\mu.$
    	\If{$\mu \le c$} set $u \leftarrow v$, and continue (i.e., return to line 3).
    	\EndIf
    \EndFor
\EndFor
\end{algorithmic}
\label{alg:searchRL}
\end{algorithm}

\begin{lemma}
\label{mainlemma}
Algorithm \ref{alg:searchRL} runs in randomized log-space, has $O(\log nk)$ influential bits, and with high probability it outputs a path from $s$ to $t$ in expected polynomial time.
\end{lemma}

\begin{proof}
We first show the algorithm runs in randomized log-space. Then we show the algorithm outputs a path with high probability in polynomial time, and then we show the output has $O(\log nk)$ influential bits. \\

\noindent \textbf{Runs in randomized log-space:} At every point in the algorithm, we must store in memory the value of $c$ (which requires $\log (\poly(n, k)) = O( \log nk)$ bits), the current value of $d$, which requires $\log k$ bits, and the current vertex $u$, which requires $\log n$ bits. In addition, in line 6 we estimate the value of $p_{d-1}(v, t)$, which can be done in log-space by Lemma \ref{estprob}. Hence, the total number of bits needed is $O( \log nk)$, which is logarithmic in the input size. \\

\noindent \textbf{With high probability outputs a path from $s$ to $t$ in polynomial time:} Out of the possible values for $c$ in the set $\{\frac{1}{k^4n^4}, \frac{2}{k^4n^4}, \dots, \frac{k^2n^2}{k^4n^4} \}$, at most $kn$ of them could satisfy $|1/2 - p_i(v, t) - c| \le \frac{1}{n^5k^5}$ for some value of $1 \le i \le k$ and vertex $v$ (since there are $kn$ possible values of $p_i(v, t)$). We choose such a value with probability at most $\frac{1}{kn}$ (since there are at most $kn$ such ``bad" choices, out of $k^2n^2$ total choices for $c$). Now, consider the other values of $c$, which do not satisfy $|1/2 - p_i(v, t) - c| \le \frac{1}{n^5k^5}$ for any $i$ and $v$. We now show that with high probability if the if statement in line 7 is satisfied, it is the case that with high probability $1/2 - p_{d-1}(v, t) \le c$. This is since $1/2 - p_{d-1}(v, t)$ is more than $\frac{1}{k^5n^5}$ away from $c$, and by Lemma \ref{estprob}, the estimate for $p_{d-1}(v, t)$ is within $\frac{1}{k^5n^5}$ of the true value of $p_{d-1}(v, t)$ with high probability.
Since with high probability $p_d(u, t) \ge 1/2 - c$, vertex $u$ must have an outneighbor $v$ satisfying $p_{d-1}(v, t) \ge 1/2 - c$. Since $p_{d-1}(v, t)$ is further than $\frac{1}{n^5k^5}$ from $c$, by Lemma \ref{estprob}, with high probability when reaching the vertex $v$ in the for loop of line 5, in line $7$ the if statement will be satisfied. Hence, with high probability, $u$ will change on each iteration of the for loop of line 3, and we maintain that with high probability throughout the algorithm the values of $d$ and $u$  satisfy $p_d(u, t) \ge 1/2 - c$.

Now we show that the algorithm succeeds with high probability. Once again, the probability we choose a $c$ satisfying $|1/2 - p_i(v, t) - c| \le \frac{1}{n^5k^5}$ for some $1 \le i \le k$ and vertex $v$ is at most $\frac{1}{nk}.$ The remaining probabilistic parts of the algorithm come from estimating $1/2 - p_i(u, t)$ to an additive error of $\frac{1}{n^5k^5}.$ By Lemma \ref{estprob}, this has error probability at most $2e^{-2nk}$ per estimate. As we make at most $nk$ estimates, the error probability here is bounded by $2nke^{-2nk}$, which is low. \\

\noindent \textbf{Output has $O(\log nk)$ influential bits:} We claim the influential bits used by the algorithm are the bits used to pick $c$. Out of the values $c$ in the set $\{\frac{1}{k^4n^4}, \frac{2}{k^4n^4}, \dots, \frac{k^2n^2}{k^4n^4} \}$, at most $kn$ of them could satisfy $|1/2 - p_i(v, t) - c| \le \frac{1}{n^5k^5}$ for some values of $1 \le i \le k$ and vertices $v.$ The probability that we pick such a $c$ is $\frac{nk}{n^2 k^2} = \frac{1}{nk}$, so with high probability we do not pick such a $c$.

Now, for the remaining values of $c$, the algorithm will have the same output with high probability over the remaining random bits. We note that the only other place where randomness is used is in line 6 to estimate $\mu$. Note that if $1/2 - p_{d-1}(v, t) \le c$, we also know that $1/2 - p_{d-1}(v, t) \le c - \frac{1}{n^5k^5}$. Hence, by Lemma \ref{estprob}, the probability that in this case the if statement in line 7 is not satisfied is at most $2 e^{-2nk}$. Similarly, if $1/2 - p_{d-1}(v, t) \ge c$, we know that $1/2 - p_{d-1}(v, t) \ge c+\frac{1}{n^5k^5}$, and so the if statement in line 7 is satisfied with probability at most $2 e^{-2nk}$. Hence, with high probability (at least $1 - 2nke^{-2nk}$) the if statement in line $7$ is satisfied if and only if $1/2 - p_{d-1}(v, t) \le c$, and so the output is the same for almost all choices of the remaining random bits.
\end{proof}

Lemma \ref{mainlemma} immediately implies Theorem \ref{psdlintro}, completing the proof.

\subsection{Why we cannot try all possible thresholds}

One idea to make the algorithm pseudo-deterministic would be to try every possible value of $c$ (of which there are polynomially many, and therefore can be enumerated), thus removing the randomization required to sample $c$. This idea will not immediately provide a pseudo-deterministic algorithm, as we explain below.

Consider the approach of going over all possible values of $c$ in some set and choosing the first ``good one'', i.e. the first value of $c$ which is far from all values of $1/2 - p_i(v, t)$. The problem with such an algorithm is that for a fixed value of $c$, it may be hard to tell whether it is a ``good'' value of $c$. Suppose, for example, that we call a value ``good'' if it is at distance at least $1/n^2k^2$ from any value of $1/2 - p_i(v, t)$. Then, if the distance is exactly $1/n^2k^2$, it is not clear how one can check if the value of $c$ is good or not. If we simply estimate the values of $1/2 - p_i(v, t)$ and see if out estimates are at distance at least $1/n^2k^2$, we will sometimes choose $c$, and sometimes we will not (depending on the randomness we use to test whether $c$ is good). Hence, the algorithm will not be pseudo-deterministic, since this value of $c$ will sometimes be chosen, and sometimes a different value of $c$ will be chosen.

\subsection{Discussion of Algorithm \ref{alg:searchRL}}

\label{entropy}

Algorithm \ref{alg:searchRL} has the property that its output, when viewed as a distribution depending on the random bits chosen by the algorithm, has entropy $O(\log n).$ This essentially follows from the fact that the output of Algorithm \ref{alg:searchRL} with very high probability depends on only its first $O(\log n)$ random bits. Hence, after amplifying the success probability, one can show the entropy of the output would be $O(\log n).$ We note that for a pseudo-deterministic algorithm, the output has entropy less than $1$. An arbitrary search-RL algorithm can have polynomial entropy.

Another way to view Algorithm \ref{alg:searchRL} is that with high probability, the output will be one of polynomially many options (as opposed to a unique option, which would be achieved by a pseudo-deterministic algorithm). That is, for each input $x$, there exists a list $L_x$ of polynomial size such that with high probability, the output is in $L_x$. This follows from the fact that with high probability, the output of Algorithm \ref{alg:searchRL} only depends on its first $O(\log n)$ random bits. Therefore, with high probability, the output will be one of $2^{O(\log n)} = \poly(n)$ different paths. Another way to see this is that with high probability the outputted path depends only on the choice of $c$, and there are polynomially many ($n^2 k^2$) possible values for $c$.

\section{Improved Pseudo-deterministic Algorithms for Connectivity}

In this section, we show faster pseudo-deterministic algorithms for both undirected connectivity and directed connectivity in Eulerian graphs. While both of these problems have been shown to be in deterministic log-space \cite{reingold, RTV}, our algorithms here have a much lower run-time than those in \cite{reingold, RTV}.

We note that throughout this section, to compute runtime, instead of dealing with Turing machines, we assume that we can make the following queries in $O(1)$ time: for a vertex $v$ we can query the degree of vertex $v$, and given a vertex $v$ and an integer $i$ we can query the $i$-th neighbor of $v$ (if  $v$ has fewer than $i$ neighbors, such a query returns $\bot$). In the case of an Eulerian graph, we assume that for a vertex $v$, we can query the degree of $v$, its $i$-th in-neighbor, and its $i$-th out-neighbor.

\subsection{Undirected Graphs}\label{subsec:undirected}

In this section, we present the algorithm for undirected graphs. Throughout we assume that we have a graph $G$ (possibly with multiedges or self-loops) with $n$ vertices and $m$ edges, and we wish to find a path from vertex $s$ to vertex $t$ (assuming such a path exists). We will number the vertices from $1$ to $n$, and refer to the $k^{th}$ vertex as ``vertex $k$''. The idea for the algorithm is as follows. First, note that checking connectivity \emph{using randomness} in undirected graphs is possible \cite{undirectedRL}: if vertices $s$ and $t$ are connected, then a random walk starting from $s$ of length $\tilde{O}(mn)$ will reach reach $t$ with high probability:

\begin{lemma}
\label{randconnect}
Given an undirected or Eulerian graph $G$ and two vertices $s$ and $t$, there exists a randomized algorithm running in time $\tilde{O}(mn)$ that checks whether there exists a path from $s$ to $t$, and succeeds with probability $1 - \frac{1}{n^{10}}.$
\end{lemma}
A version of Lemma \ref{randconnect} has been shown in \cite{undirectedRL}. For completeness, we include a proof in Appendix \ref{app:randconnect}.

Now, we proceed to prove Theorem \ref{slpsd}, restated here for convenience.

\begin{theorem*}[Pseudo-deterministic Undirected Connectivity in $\tilde{O}(mn^3)$ time, $O(\log n)$ space]

Let $G$ be a given undirected graph with $n$ vertices and $m$ edges. Given two vertices $s$ and $t$ of $G$ which are connected, there is a pseudo-deterministic log-space algorithm which outputs a path from $s$ to $t$. Furthermore, the algorithm runs in time $\tilde{O}(mn^3)$.
\end{theorem*}

For our pseudo-deterministic algorithm for undirected connectivity, we use the following approach. We delete vertices of small ID from the graph one at a time (excluding $s$ and $t$), and check if vertex $s$ is still connected to $t$. Now, suppose that after deleting vertices $1, 2, 3, \dots, k-1$, excluding $s$ and $t$, (recall that we number the vertices from $1$ to $n$, and refer to the $i$th vertex as ``vertex $i$"), $s$ is still connected to $t$. However, suppose that when we delete vertices $1, 2, \ldots, k$, vertex $s$ is no longer connected to $t$. Then we will recursively find paths $s \to k$ and $k \to t$. Repeating this process, we will get a path from $s \to t.$

Of course, as described, this algorithm will not run in log-space (since, for example, one must store in memory which vertices have been removed, as well as the recursion tree, both of which may require large space). With a few modifications though, one can adapt the algorithm to run in log-space. For a complete description of the algorithm, see Algorithm \ref{alg:psdSL}.

We use the variables $v_{cur}$ and $v_{dest}$ to denote the vertex our walk is currently on and the vertex which is the destination (at the current level of the recursion).

\begin{algorithm}
\caption{Pseudo-deterministic log-space algorithm for undirected connectivity.}
\begin{algorithmic}[1]
%\Statex
\State Use Lemma \ref{randconnect} to test if $s$ and $t$ are connected. If they are not, return ``not connected''
\State Set $v_{cur} = s$.
\While{$v_{cur}\neq t$}
	\State Set $v_{dest} = t$
	\For{$k = 1, 2, \dots, n$}
    	\If{$v_{cur}$ is adjacent to $v_{dest}$} set $v_{cur} \leftarrow v_{dest}$, print $v_{cur}$ (on the output tape), and go to line 3.
        \EndIf
        \If{$v_{cur}$ is not connected to $v_{dest}$ in the graph with vertices $v_{cur}$, $v_{dest}$, and $k+1, k+2, \dots, n$ (for a detailed description of the implementation of this step, see the proof of Lemma \ref{undiralg})} set $v_{dest} \leftarrow k.$
        \EndIf
	\EndFor
\EndWhile
\end{algorithmic}
\label{alg:psdSL}
\end{algorithm}

\begin{lemma} \label{undiralg} Given a graph $G$, Algorithm \ref{alg:psdSL} outputs a path from vertex $s$ to $t$ with high probability (if such a path exists), and runs in pseudo-deterministic log-space and time $\tilde{O}(mn^3)$.
\end{lemma}

\begin{proof}
We begin by providing a more detailed description of the implementation of line 7. Then, we will analyze the algorithm in detail. Specifically, we will show that the algorithm returns a path from $s$ to $t$ with high probability, uses logarithmic space, runs in time $\tilde{O}(mn^3)$, and is pseudo-deterministic.\\

\noindent \textbf{Description of line 7:} In order to check if $v_{cur}$ and $v_{dest}$ are connected, we run a random walk on the graph $H$ with vertices $v_{cur}$, $v_{dest}$, and $k+1, k+2, \dots, n$. To do so, in each step of the random walk, if the walk is currently on $v$, we pick a random edge $(v, u)$ adjacent to $v$, and test if the other endpoint $u$ of the edge is in $H$ (this can be done by testing if the ID of $u$ is larger or smaller than $k$). If it is, the random walk proceeds to $u$. Otherwise, the random walk remains at $v$. To analyze the runtime of this walk, we note that such a walk is identical to a random walk on the graph $H$ which is the graph induced by $G$ on the vertices $v_{cur}$, $v_{dest}$, and $k+1, k+2, \dots, n$, along with self loops, where every edge $(v, u)$ where $v \in H$ and $u \notin H$ is replaced by a self loop at $v$. Since this graph has fewer than $n$ vertices, and at most $m$ edges, by Lemma \ref{randconnect} Line 7 takes time $\tilde{O}(mn).$\\

\noindent\textbf{Returns a path from $s$ to $t$ with high probability:} The key claim is that the variable $v_{cur}$ never returns to the same vertex twice, and changes in each iteration of the while loop in line 3. This implies the success of the algorithm since then after at most $n$ iterations of line 3, $v_{cur}$ must have achieved the value of $t$ at some point.

To prove that $v_{cur}$ never returns to the same vertex twice, and changes in each iteration of the while loop, we consider the ``destination sequence'' of the vertex $v = v_{cur}$. We define the destination sequence of $v$ to be the sequence of values the $v_{dest}$ variable takes during the period when $v = v_{cur}$.  That is, the destination sequence of some vertex $v = v_{cur}$ is $(t, c_1, c_2, \ldots, c_i)$ where $c_j$ is the value of $k$ on the $j$th time that the if statement in line 7 evaluated to True. Note that the destination sequence is a function of a vertex (i.e., the sequence of values the $v_{dest}$ variable takes during the period when $v = v_{cur}$ depends only on $v$, assuming line 7 was implemented successfully). It's worth noting that since during the period that $v_{cur} = v$, the value of $k$ only increases, so for a destination sequence $(t, c_1, c_2, \ldots, c_i)$ we have $c_1 < c_2 < \ldots < c_i$.

We note that $c_{j+1}$ is the smallest integer such on the graph $G$ induces on the vertices $v_{cur}, c_{j}$, and $c_{j+1} + 1, c_{j+1} + 2, c_{j+1}, \ldots, n$, the vertex $v_{cur}$ is not connected to $c_{j}$

Consider the following total ordering on sequences. A sequence $C = (t, c_1, c_2, \dots, c_i)$ is larger than $D = (t, d_1, d_2, \dots, d_i, \dots, d_j)$ if either $c_\ell = d_\ell$ for all $1 \le  \ell \le i$ (and $i < j$), or for the first value of $\ell$ for which $c_\ell \neq d_{\ell}$, we have $c_\ell > d_{\ell}$. Otherwise, if $C \neq D$, we say that $C < D.$

We claim that during the algorithm, the destination sequences strictly increase according to the above ordering whenever the value of $v_{cur}$ changes, which happens every time Algorithm \ref{alg:psdSL} returns to Line 3. Note that this would imply that $v_{cur}$ never achieves the same vertex $v$ twice (if it does, that contradicts the fact that the destination sequence of $v_{cur}$ must have increased). Hence, it will suffice to show that the destination sequence of $v_{cur}$ increases according to the above ordering whenever Algorithm \ref{alg:psdSL} returns to Line 3.

Suppose that $u = v_{cur}$. Let the destination sequence of $u$ be $(t, c_1, c_2, \dots, c_i, u')$, where the next value achieved by $v_{cur}$ is $u'$. By construction, we have that $c_1 < c_2 < \dots < c_i < u'.$ Also, let the destination sequence of $u'$ be $(t, d_1, d_2, \dots, d_i, \dots, d_j)$, where we similarly have that $d_1 < d_2 < \dots < d_j.$ We will show that the destination sequence of $u'$ is larger than that of $u$ under the ordering defined above.

First, we claim that $c_\ell = d_\ell$ for $1 \le \ell \le i.$ We can show this by induction. We first show that $c_1 = d_1.$ Indeed, we have that $d_1 \le c_1$ because deleting vertices $1, 2, \dots, c_1$ (which doesn't include $u'$) from the graph will disconnect $u'$ from $t$, as $u$ and $u'$ are adjacent. To show that $d_1 \ge c_1$, we show that there is a path from $u'$ to $t$ that doesn't use any vertices (other than maybe $t$) with labels less than $c_1$. Indeed, deleting vertices $1, 2, \dots, u'$ disconnected $u$ from $c_i$, but deleting $1, 2, \dots, u' - 1$ didn't, so there is a path from $u'$ to $c_i$ that doesn't use any vertices (other than $c_i$) with labels less that $u'.$ Similarly, for any $2 \le j \le i$, deleting vertices $1, 2, \dots, c_j$ disconnected $u$ from $c_{j-1}$, but deleting $1, 2, \dots, c_j - 1$ didn't. Therefore, there is a path from $c_j$ to $c_{j-1}$ that doesn't use any vertices (other than $c_{j-1}$) with labels less than $c_j.$ Finally, there is a path from $c_1$ to $t$ that doesn't use any vertices (other than maybe $t$) with labels less than $c_1$. By merging all these paths together at the endpoints, we get a path from $u'$ to $t$ doesn't use any vertices with labels less than $c_1$, as desired. This implies $c_1 \le d_1$. Combining this with $c_1 \ge d_1$ which we showed above, we now have $c_1 = d_1$. Now, for some integer $p < i$, assume by induction that $c_\ell = d_\ell$ for all $1 \le \ell \le p$. We want to show that $c_{p+1} = d_{p+1}.$ This can be shown using the exact same argument for showing that $c_1 = d_1$, with $t$ replaced by $c_p$. Therefore, we have that $d_{p+1} = c_{p+1}.$ Now, by the same argument again (with $t$ replaced by $c_i$), we can show that either  $d_{i+1} > u'$, or $d_{i+1}$ doesn't exist. The former occurs when $u'$ and $c_i$ aren't adjacent, and the latter happens when they are. Thus, the destination sequence of $u'$ is larger than that of $u$ under the ordering: all the first $i$ entries stay the same, and we either delete the last entry, or make it larger. Thus, $v_{cur}$ never returns to the same vertex. Hence, it must eventually reach $t$, proving that the algorithm returns a path from $s$ to $t$.

To show that the algorithm succeeds with high probability, note that the only lines which are probabilistic are lines 1 and 7. For each execution of these lines, it has failure probability at most $\frac{1}{n^{10}}$ by Lemma \ref{randconnect}. Each of these lines is run at most $n^{3}$ times, so the total failure probability is bounded by $O\left(\frac{1}{n^7}\right)$.\\

\noindent \textbf{Uses $O(\log n)$ space:} At every point in the algorithm, the following is stored: $v_{cur}$, $v_{dest}$, and $k$ (all of which require space $O( \log n )$). In addition, in line 7 the algorithm will run a random walk, which will require storing the id of the current vertex, as well as a counter storing how many steps of the random walk have been executed. Both of these can be stored using logarithmic space.\\

\noindent \textbf{Runs in time $\tilde{O}(mn^3)$:} $v_{cur}$ can take at most $n$ different values and with high probability does not take the same value more than once (see the paragraph above on why the algorithm returns a path with high probability for a proof of this fact). Since $v_{cur}$ changes its value in each iteration of the while loop, line 3 executes at most $n$ times. Line $5$ runs at most $n$ times. Line 6 is checkable in $O(\log n)$ time, and line 7 runs in time $\tilde{O}(mn)$ as shown earlier in the proof (as part of the description of line 7). Therefore, the total runtime is $\tilde{O}(n \times n \times mn) = \tilde{O}(mn^3)$, as desired. \\

\noindent \textbf{Is pseudo-deterministic:} Randomness is only used in lines 6 and 7, and this is only for checking connectivity. Testing connectivity is a pseudo-deterministic protocol since given two vertices, with high probability when testing for connectivity twice, the same result will be output (namely, if the two vertices are connected, with high probability the algorithm will output that they are connected in both runs. If the two vertices are not connected, with high probability the algorithm will output that they are not connected in both runs). Since all uses of randomization is for checking connectivity in a pseudo-deterministic fashion, the algorithm as a whole is pseudo-deterministic.
\end{proof}

\subsection{Eulerian Graphs}
\label{sec:euler}

In this section, we show an efficient pseudo-deterministic log-space algorithm for finding paths in Eulerian graphs (directed graphs such that for every vertex $v$, the indegree and outdegree of $v$ are equal). Recall that in our model of computation, for a vertex $v$, we can query either the degree of $v$, the $i$-th in-neighbor of $v$, or the $i$-th out-neighbor of $v$ in $O(1)$ time. We will prove Theorem \ref{eulerpsd}, repeated below for convenience:

\begin{theorem*}[Connectivity in Eulerian graphs in in $\tilde{O}(m^5n^3)$ time, $O(\log n)$ space]
Given an Eulerian graph $G$ and two vertices $s$ and $t$ where there exists a path from $s$ to $t$, there is a pseudo-deterministic log-space algorithm which outputs a path from $s$ to $t$. Furthermore, the algorithm runs in time $\tilde{O}(m^5n^3)$.
\end{theorem*}

The algorithm will be a variation on the algorithm for undirected graphs of Subsection \ref{subsec:undirected}. 

First, note that as in the case with undirected graphs, checking connectivity in Eulerian graphs can be done efficiently using a randomized algorithm (see Lemma \ref{randconnect}). We first would like to note that the algorithm for undirected graphs doesn't immediately generalize to the Eulerian case. This is because the algorithm for undirected graphs involves checking for connectivity on the graph $G$ with some vertices (and their adjacent edges) removed. The reason it is possible to check connectivity in this modified graph is that after removing vertices, the resulting graph is still undirected. However, in the Eulerian case, removing vertices along with their edges may result in a non-Eulerian graph.  So, instead of removing vertices from the graph, we instead remove directed cycles in the graph. One of the key observations is that when deleting a cycle, the resulting graph is still Eulerian, so we can apply Lemma \ref{randconnect} to test for connectivity.

At the high level, the algorithm proceeds as follows: we remove directed cycles from the graph and check (using randomization) whether vertex $t$ can be reached from vertex $s$. If it can, we continue removing cycles. If not, then we recursively try to go from vertex $s$ to some vertex on the cycle whose deletion disconnects vertex $s$ and $t$ (call this cycle $C$). After we find a path to some vertex on the cycle $C$, note that there exists a vertex $v$ on $C$ such that deleting $C$ does not disconnect $v$ from $t$. So we then walk on the cycle to $v$, and then recursively apply the algorithm to find a path from $v$ to $t$.

As described, this algorithm will not work in log-space, because it is not clear how the algorithm can store in memory a description of which cycles have been deleted. The following lemma provides us with a way to delete cycles in a specified way, so we can compute in log-space whether an edge is part of a deleted cycle or not.

\begin{lemma}
\label{edgebiject}
Let $E$ be the set of edges of $G$. There exists a log-space-computable permutation $f: E \to E$ that satisfies the following property: if $e$ is an inedge of vertex $v$, then $f(e)$ is an outedge of $v$. In particular, this condition implies for any edge $e$, we have that $e, f(e), f^2(e), \dots$ forms a cycle in $G$.
\end{lemma}

\begin{proof}
Take a vertex $v$ of indegree $d$, and take some ordering of its inedges $e^\text{in}_1, e^\text{in}_2, \dots, e^\text{in}_d$ (say, in lexicographic order), and some ordering of the outedges $e^\text{out}_1, e^\text{out}_2, \dots, e^\text{out}_d$ (say again, in lexicographic order). Then simply set $f(e^\text{in}_i) = e^\text{out}_i$.
\end{proof}

Note that in our model, computing $f(e)$ takes $O(n)$ time for each edge $e$. In particular, for a cycle $C_i$ formed by repeatedly applying $f$ to some edges, and an edge $e$ in the cycle, computing the next edge in the cycle takes time $O(n)$.

The importance of the lemma is that it provides us with a way to delete cycles in some order: we begin from the ``smallest'' edge (in whatever ordering) and delete the cycle associated with that edge. Then, we pick the second smallest edge, and delete its cycle, etc. When executing this algorithm, we may try to delete a cycle multiple times, since multiple edges correspond to the same cycle, but this will not be an issue.

See Algorithm \ref{alg:euler} for a precise description of the algorithm. As in the undirected case, we use the variables $v_{cur}$ and $v_{dest}$ to denote the vertex our walk is currently on and the current destination. We let the set of $e_i$ be the edges in $G$ (we denote the set of edges of $G$ as $E$), and let $C_i$ be the cycle $(e_i, f(e_i), f^2(e_i), \dots, e_i)$ in $G$. We note that it is possible for $C_i$ and $C_j$ to have the same set of edges for $i \neq j$ (this will not affect the correctness of the algorithm).

\begin{algorithm}
\caption{Pseudo-deterministic log-space algorithm for connectivity in Eulerian digraphs.}
\begin{algorithmic}[1]

%\Statex

\State Set $v_{cur} = s.$ Write $v_{cur}$ on the output tape.
\While{$v_{cur} \neq t$}
	\State Set $v_{dest} = t$.
    \For{$k = 1 \dots m$}
    	\If{$v_{cur}$ and $v_{dest}$ are not connected using only edges in $E \setminus \{C_1, \dots, C_k\}$ (more details of the implementation of this step are in the body of the paper in the proof of Lemma \ref{euleralg})}
        	\If{$v_{cur} \in C_k$}
            	\State Find a vertex $v \in C_k$ such that $v$ and $v_{dest}$ are connected using edges only in $E \setminus \{C_1, \dots, C_k\}$.
                \State Walk on $C_k$ from $v_{cur}$ until you reach $v$, print the vertices on this path (on the output tape).
                \State Set $v_{cur} \leftarrow v$, return to top of while loop.
            \Else
            	\State Find a vertex $v \in C_k$ such that $v_{cur}$ can get to $v$ in $G \setminus \{C_1, \dots, C_k\}$
            	\State Set $v_{dest} \leftarrow v.$
            \EndIf
        \EndIf
    \EndFor
\EndWhile

\end{algorithmic}
\label{alg:euler}
\end{algorithm}

\begin{lemma}
\label{euleralg}
Algorithm \ref{alg:euler} runs in time $\tilde{O}(m^5n^3)$, is pseudo-deterministic, uses logarithmic space, and outputs a path from $s$ to 
$t$ with high probability.
\end{lemma}

\begin{proof} We first give a more detailed description of the implementation of line 5. Then, we analyze the algorithm in steps, first showing that it returns a path with high probability, and then showing that it runs in pseudo-deterministic log-space with runtime $\tilde{O}(m^5n^3)$. The proof closely follows the approach of the proof in the undirected case.\\

\noindent \textbf{Implementation of line 5:} Below, we describe the details of the implementation of step 5. As done in the proof of Lemma \ref{randconnect} in Appendix \ref{app:randconnect}, we check connectivity by making every edge in the graph $G$ undirected, and then perform a random walk on the undirected graph $G$. The key difficulty is to ensure that we can check if an edge is in one of the deleted cycles efficiently in log-space.

Say that we are on vertex $u$, and the randomly chosen neighbor which is next in the random walk is $v$. Let $e$ be the edge between $u$ and $v$. We wish to check whether $e$ is in any of the cycles $C_1, C_2, \dots, C_k$. To check whether edge $e$ is on the cycle $C_i$, we can  check whether $e$ is any of the edges $e_i, f(e_i), f^2(e_i), \dots e_i$, which can all be computed in log-space. To see if $e$ is on any of the cycles $C_1, \dots, C_k$, we check if $e$ is in each of the $C_i$. Each such check takes time $O(mn)$ (since the cycle is of length $O(m)$, and given an edge $e$, computing the next edge in the cycle takes time $O(n)$). Hence, in total it takes time $O(kmn)$.

Now, if $e$ is on one of the cycles, the random walk stays at $u$, and otherwise the random walk proceeds to $v$. It is clear that this is equivalent to taking a random walk on the graph $G'$, where $G'$ is the graph $G$ but with all edges $(u', v')$ in at least one of the cycles $C_1, \dots, C_k$ replaced with a self-loop at $u'$ (since, if such an edge is chosen, the random walk stays at $u'$. As $G'$ still has at most $m$ edges and $n$ vertices, checking connectivity takes $\tilde{O}(mn)$ time by Lemma \ref{randconnect}.\\

\noindent \textbf{Returns a path with high probability:} As in the undirected case, the main claim is that $v_{cur}$ never repeats a vertex on two different iterations of the while loop of line 2. To prove that $v_{cur}$ is never repeated, we will use the notion of the ``destination sequence'', similar to the undirected case. We note that our definition here of a destination sequence is different from the definition in the undirected case. We say that the associated destination sequence to $v_{cur}$ is the sequence of cycles $(C_{i_1}, C_{i_2}, \dots, C_{i_k})$, where we add $C_{i_j}$ to the sequence if the if statement of step 5 of the algorithm was true when $k = i_j$. Note that this implies that $i_1 \le i_2 \le \dots \le i_k.$

Now, as in the proof of Lemma \ref{undiralg}, we give a total ordering on all destination sequences. Consider two destination sequences $i = (C_{i_1}, C_{i_2}, \dots, C_{i_k})$ and $j = (C_{j_1}, C_{j_2}, \dots, C_{j_k}, \dots, C_{j_m}).$ Say that $i$ is greater than $j$ if either $i_\ell = j_\ell$ for all $1 \le \ell \le k$ (and $k < m$), or for the smallest value of $\ell$ such that $i_\ell \neq j_\ell$, we have that $i_\ell > j_\ell.$ Otherwise, if $i \neq j$, then say that $i < j$.

Now, we proceed to prove that $v_{cur}$ never repeats a value. Say that $v_{cur}$ is set to $v_{cur}'$ after one loop of line 2. Let the destination sequence of $v_{cur}$ be $i = (C_{i_1}, C_{i_2}, \dots, C_{i_k})$ and let the corresponding destination sequence of $v_{cur}'$ be $(C_{j_1}, C_{j_2}, \dots, C_{j_m})$. First, we claim that $i_p = j_p$ for all $1 \le p \le k-1.$ This is because $v_{cur}$ is connected to $v_{cur}'$ via cycle $C_{i_k}$, so deleting $C_{i_p}$ disconnects $v_{cur}$ and $v_{dest}$ if and only if it disconnects $v_{cur}'$ and $v_{dest}.$ As in the proof of Lemma \ref{undiralg}, we have two situations now. One case is that $v_{cur}' \in C_{i_{k-1}}$, and therefore, $m = k-1$ (the destination sequence for $v_{cur}'$ is one shorter than that of $v_{cur}$). The other is that $v_{cur}' \not\in C_{i_{k-1}}$, and therefore, $j_k \ge i_k$, as deleting $C_{i_k}$ doesn't disconnect $v_{cur}'$ and $v_{dest}$ by the condition of line $7$ of the algorithm. So the destination sequence of $v_{cur}'$ is greater than that of $v_{cur}$ under the total ordering described above, which implies that $v_{cur}$ can never repeat a vertex.

To see that the algorithm succeeds with high probability, note that the only randomness is in lines 5, 7, and 11 for checking connectivity between two vertices. We will check connectivity at most $O(nm^2)$ times, so the failure probability is bounded by $\frac{nm^2}{n^{10}}$ by Lemma \ref{randconnect}, as desired.\\

\noindent \textbf{Uses $O(\log n)$ space:} The information our algorithm needs to store is: $s, t, v_{cur}, v_{dest}$, and $k$. After that, by Lemma \ref{edgebiject}, we can compute whether an edge $e$ is part of a cycle $C_k$ is log-space, and testing whether two vertices are connected in an Eulerian graph can be done in (randomized) log-space by Lemma \ref{randconnect}. Therefore, everything can be implemented in log-space. \\

\noindent \textbf{Runs in time $\tilde{O}(m^5n^3)$:} Line 2 repeats at most $n$ times since $v_{cur}$ never repeats, and Line 4 repeats at most $m$ times since there are $m$ possible values for $k$. Due to Lemma \ref{randconnect}, each execution of Line 5 takes $\tilde{O}(mn) \times O(m^2 n)$ time, where the $O(m^2n)$ comes from having to check whether each edge we try to use comes from one of the cycles $C_1, \dots, C_k$ (a factor $m$ from the fact that there are up to $m$ cycles $C_i$, a factor $m$ from the fact that the size of each $C_i$ is at most $m$, and a factor $O(n)$ because given some edge $e$ in $C_i$, computing the next edge in the cycle takes time $O(n)$).  Line 7 and 11 take time $O(m) \times \tilde{O}(mn) \times O(m^2n)$, for the same reason as above, except with the extra $O(m)$ factor for having to check all vertices on the cycle $C_k.$ Therefore, our runtime bound is $O(n) \times O(m) \times (\tilde{O}(mn) \times O(m^2) + O(m) \times \tilde{O}(mn) \times O(m^2n)) = \tilde{O}(m^5n^3)$, as desired. \\

\noindent \textbf{Is pseudo-deterministic:} The only randomness is used to check connectivity between pairs of vertices. As each of these checks succeeds with high probability, this clearly implies that our randomness used will not affect the output of the algorithm, since if the two vertices tested are connected, with high probability the same result (of ``accept") will be outputted, and if the two vertices tested are not connected, with high probability the same result (of ``reject") will be outputted.
\end{proof}

\section{Discussion}
\label{discussion}
The main problem left open is that of search-RL vs pseudo-deterministic-L:
\begin{prob} \label{rlpsdl} Can every problem in search-RL be solved pseudo-deterministically in RL?\end{prob}

A notable open problem in complexity is whether $NL$ equals $UL$. It is known that under randomized reductions, with \textit{two way access} to the random bits, $NL$ is reducible to $UL$ (in fact, it is reducible to $UL \cap coUL$) \cite{RA}. It is not known whether $NL$ is reducible to $UL$ when given one-way access to the random bits. A reproducible reduction from $NL$ to $UL$ would imply such a result, giving us the following problem:

\begin{prob} Does there exist a reproducible log-space reduction from NL to UL?
\end{prob}

Another interesting problem would be to fully derandomize the pseudo-deterministic algorithms we present for undirected and Eulerian connectivity, in order to get deterministic log-space algorithms which work in low polynomial time.
\begin{prob}
Does there exists a deterministic log-space algorithm for undirected connectivity (or connectivity in Eulerian graphs) using low time complexity?
\end{prob}

There are several natural extensions of the notion of reproducibility to the time-bounded setting, some which may be worth exploring. A noteworthy extension is that of low-entropy output algorithms. Our algorithm for search-RL has the property that its output, when viewed as a random variable depending on the random choices of the algorithm, has $O(\log n)$ entropy. It may be interesting to understand such algorithms in the context of time-bounded computation.
\begin{prob}
Let search-BPP($\log n$) be the set of problems solvable by randomized polynomial time machines, whose outputs (when viewed as random variables depending on the random choices of the algorithms) have $O(\log n)$ entropy. What is relationship between search-BPP($\log n$) and search-BPP? What is the relationship between search-BPP($\log n$) and pseudo-deterministic-BPP?
\end{prob}

\section*{Acknowledgments}
Thanks to Shafi Goldwasser for discussions, and for many helpful comments on older versions of the paper. Many thanks to Omer Reingold for valuable discussions, especially those leading up to Subsection \ref{sec:euler} on Eulerian graphs. Thanks to Uri Feige, Dhiraj Holden, and Aleksander Madry for discussions. 
\bibliographystyle{plain}
\bibliography{bibfile}

\begin{thebibliography}{10}

\bibitem{undirectedRL}
Romas Aleliunas, Richard~M Karp, Richard~J Lipton, Laszlo Lovasz, and Charles
  Rackoff.
\newblock Random walks, universal traversal sequences, and the complexity of
  maze problems.
\newblock In {\em Foundations of Computer Science, 1979., 20th Annual Symposium
  on}, pages 218--223. IEEE, 1979.

\bibitem{GW}
Anna G{\'a}l and Avi Wigderson.
\newblock Boolean complexity classes vs. their arithmetic analogs.
\newblock {\em Random Structures and Algorithms}, 9(1-2):99--111, 1996.

\bibitem{GG}
Eran Gat and Shafi Goldwasser.
\newblock Probabilistic search algorithms with unique answers and their
  cryptographic applications.
\newblock In {\em Electronic Colloquium on Computational Complexity (ECCC)},
  volume~18, page 136, 2011.

\bibitem{GGR}
Oded Goldreich, Shafi Goldwasser, and Dana Ron.
\newblock On the possibilities and limitations of pseudodeterministic
  algorithms.
\newblock In {\em Proceedings of the 4th conference on Innovations in
  Theoretical Computer Science}, pages 127--138. ACM, 2013.

\bibitem{matching}
Shafi Goldwasser and Ofer Grossman.
\newblock Perfect bipartite matching in pseudo-deterministic {RNC}.
\newblock In {\em Electronic Colloquium on Computational Complexity (ECCC)},
  volume~22, page 208, 2015.

\bibitem{proofs}
Shafi Goldwasser, Ofer Grossman, and Dhiraj Holden.
\newblock Pseudo-deterministic proofs.
\newblock {\em arXiv preprint arXiv:1706.04641}, 2017.

\bibitem{roots}
Ofer Grossman.
\newblock Finding primitive roots pseudo-deterministically.
\newblock In {\em Electronic Colloquium on Computational Complexity (ECCC)},
  volume~22, page 207, 2015.

\bibitem{dhiraj}
Dhiraj Holden.
\newblock A note on unconditional subexponential-time pseudo-deterministic
  algorithms for {BPP} search problems.
\newblock {\em arXiv preprint arXiv:1707.05808}, 2017.

\bibitem{OS}
Igor~C Oliveira and Rahul Santhanam.
\newblock Pseudodeterministic constructions in subexponential time.
\newblock {\em arXiv preprint arXiv:1612.01817}, 2016.

\bibitem{reingold}
Omer Reingold.
\newblock Undirected connectivity in log-space.
\newblock {\em Journal of the ACM (JACM)}, 55(4):17, 2008.

\bibitem{RTV}
Omer Reingold, Luca Trevisan, and Salil Vadhan.
\newblock Pseudorandom walks on regular digraphs and the {RL} vs. {L} problem.
\newblock In {\em Proceedings of the thirty-eighth annual ACM symposium on
  Theory of computing}, pages 457--466. ACM, 2006.

\bibitem{RA}
Klaus Reinhardt and Eric Allender.
\newblock Making nondeterminism unambiguous.
\newblock {\em SIAM Journal on Computing}, 29(4):1118--1131, 2000.

\bibitem{SZ}
Michael Saks and Shiyu Zhou.
\newblock {$BPHSPACE(S) \subseteq DSPACE (S^{3/2})$}.
\newblock {\em Journal of Computer and System Sciences}, 58(2):376--403, 1999.

\end{thebibliography}

\appendix

\section{Testing Connectivity for Undirected and Eulerian graphs in RL}\label{app:randconnect}
In this section, we prove Lemma \ref{randconnect}, repeated below for convenience:
\begin{lemma*}
Given an undirected or Eulerian graph $G$ and two vertices $s$ and $t$, there exists a randomized algorithm running in time $\tilde{O}(mn)$ that checks whether there exists a path from $s$ to $t$, and succeeds with probability $1 - \frac{1}{n^{10}}.$
\end{lemma*}
\begin{proof}
We begin by showing that for an Eulerian graph $G$, if there is an edge from vertex $u$ to vertex $v$, then there is also a path from vertex $v$ to vertex $u$. Let $V_v$ be the set of vertices reachable from $v$. Note that the number of edges incoming to $V_v$ must be the same as the number of edges going out of $V_v$. However, by the definition of $V_v$, there cannot be edges leaving the set (if there is a an edge $(v', u')$ where $v' \in V_v$ and $u' \notin V_v$, then $u'$ can be reached from $v$, and hence $u' \in V_v$, a contradiction). Hence, since there are no outgoing edges, there are also no incoming edges. Hence, since $(u, v)$ has one endpoint in $V_v$, both endpoints must be in $V_v$, so $u \in V_v$ is reachable from $v$.

Hence, in order to test reachability in Eulerian graphs, it is enough to test reachability in the undirected graph defined by making all edges of the Eulerian graph undirected. Hence, it suffices to prove the lemma for undirected graphs.

The expected number of steps needed to get to vertex $t$ after starting a random walk at vertex $s$ where there is a path from $s$ to $t$ is bounded by $2mn$ \cite{undirectedRL}. By Markov's inequality, the probability that a random walk of length $4mn$ starting at vertex $s$ doesn't reach vertex $t$ is at most $\frac{1}{2}.$ Therefore, starting at vertex $s$ and repeating $O(\log n)$ random walks of length $4mn$ provides the result. It is easy to check that in our model, taking one step in a random walk takes time $O(\log n).$
\end{proof}

\section{\textsc{Short-Walk Find Path} is complete for search-RL}\label{complete}

In this section, we prove Lemma \ref{completesearchrl}, which states that \textsc{Short-Walk Find Path} is complete for search-RL. We repeat the definition of \textsc{Short-Walk Find Path} below for convenience, and then we proceed to prove that it is complete for search-RL.

\begin{definition}[\textsc{Short-Walk Find Path}]
Let $R$ be the search problem whose valid inputs are $x = (G, s, t, 1^k)$ where $G$ is a directed graph, $s$ and $t$ are two vertices of $G$, and a random walk of length $k$ from $s$ reaches $t$ with probability at least $1 - 1/|x|$. On such an $x$, a valid output is a path of length up to $\poly(k)$ from $s$ to $t$.
\end{definition}

We now prove that \textsc{Short-Walk Find Path} is complete for search-RL. The definition of reductions in the context of search-RL is given in \cite{RTV}.

\begin{proof}
First, it is easy to see that \textsc{Short-Walk Find Path} is in search-RL, as we can just take a random walk starting from $s$ of length $k$.

Now we show that \textsc{Short-Walk Find Path} is search-RL-hard via a reduction from \textsc{Poly-Mixing Find Path}. In \cite{RTV} Section A.3 (proof of Theorem 3.1), it is shown that \textsc{Poly-Mixing Find Path} with input $(G, s, t, 1^k)$ is complete for search-RL. They also state that a path of length $m = 2k \log k$ starting from $s$ reaches $t$ with probability at least $\frac{1}{2k}$ in the problem \textsc{Poly-Mixing Find Path}. Now, we amplify this probability of $\frac{1}{2k}$ by constructing a new graph. To do this, consider a graph $G'$ which is made as follows: it has $(m+1)|V(G)|$ vertices, each of which is a pair $(i, v)$ for $0 \le i \le m$ and vertex $v \in G.$ If the edge $u \to v$ is in $G$ then add edges $(i, u) \to (i+1, v)$ for $0 \le i \le m-1$ in $G'.$ Finally, create edges $(m, v) \to (0, s)$ for all $v \neq t$, and add only a self-loop to the vertex $(m, t)$ (so if a random walk reaches $(m, t)$, the random walk will stay there forever). Then, it is easy to see that a random walk of length $\ell = m + 2(m+1) k \log x$ starting at $(0, s)$ will end at $(m, t)$ with probability at least $1 - \left(1 - \frac{1}{2k} \right)^{2k \log x} \ge 1 - \frac{1}{x}.$ Choosing $x$ larger than the length of the input gives the desired reduction. That is, when choosing such an $x$, given a solution to the \textsc{Short-Walk Find Path}, we can output a polynomially long list $y_1, y_2, \ldots, y_p$ such that at least one of the $y_i$ is a solution to the \textsc{Poly-Mixing Find Path} instance. Therefore, \textsc{Short-Walk Find Path} is complete for search-RL.
\end{proof}

\end{document}